\documentclass[reqno,10pt,a4paper,epic]{amsart}

\usepackage{hyperref}
\usepackage{amssymb}
\usepackage{url}

\allowdisplaybreaks[4]

\let\phi=\varphi
\let\kappa=\varkappa
\let\epsilon=\varepsilon

\DeclareMathOperator{\sym}{sym_{\mathrm{c}}}
\DeclareMathOperator{\cosym}{cosym_{\mathrm{c}}}

\newcommand*{\abs}[1]{\left|#1\right|}
\newcommand*{\Ev}{\mathbf{E}}
\newcommand*{\tc}{\mathrm{t}}
\newcommand*{\rd}{\mathrm{r}}
\theoremstyle{theorem}
\newtheorem{proposition}{Proposition}
\theoremstyle{definition}

\theoremstyle{remark}
\newtheorem{remark}{Remark}

\usepackage[all]{xy}
\SelectTips{cm}{}
\newcommand{\QED}{\hfill$\square$}

\usepackage{mathrsfs}
\let\mathcal\mathscr
\usepackage[mathscr]{eucal}

\newcommand{\cprime}{\/{\mathsurround=0pt$'$}}
\begin{document}

\title[Integrability properties of some symmetry reductions]{Integrability
  properties of some equations obtained by symmetry reductions}
\author{H.~Baran} \address{Mathematical Institute, Silesian University in
  Opava, Na Rybn\'{\i}\v{c}ku 1, 746 01 Opava, Czech Republic}
\email{Hynek.Baran@math.slu.cz} \author{I.S.~Krasil{\cprime}shchik}
\address{Independent University of Moscow, B. Vlasevsky 11, 119002 Moscow,
  Russia \& Mathematical Institute, Silesian University in Opava, Na
  Rybn\'{\i}\v{c}ku 1, 746 01 Opava, Czech Republic}
\email{josephkra@gmail.com} \author{O.I.~Morozov} \address{Faculty of Applied
  Mathematics, AGH University of Science and Technology, Al. Mickiewicza 30,
  Krak\'ow 30-059, Poland} \email{morozov{\symbol{64}}agh.edu.pl}
\author{P.~Voj{\v{c}}{\'{a}}k} \address{Mathematical Institute, Silesian
  University in Opava, Na Rybn\'{\i}\v{c}ku 1, 746 01 Opava, Czech Republic}
\email{Petr.Vojcak@math.slu.cz}

\date{\today}

\begin{abstract}
  In our recent paper~\cite{B-K-M-V:Sym-Red}, we gave a complete description
  of symmetry reduction of four Lax-integrable (i.e., possessing a
  zero-curvature representation with a non-removable parameter)
  $3$-dimensional equations. Here we study the behavior of the integrability
  features of the initial equations under the reduction procedure. We show
  that the ZCRs are transformed to nonlinear differential coverings of the
  resulting 2D-systems similar to the one found for the Gibbons-Tsarev
  equation in~\cite{O-S:NonHom-Sys}. Using these coverings we construct
  infinite series of (nonlocal) conservation laws and prove their
  nontriviality. We also show that the recursion operators are not preserved
  under reductions.
\end{abstract}



\maketitle

\tableofcontents
\newpage

\section*{Introduction}
\label{sec:introduction}

In~\cite{B-K-M-V:Sym-Red} we gave a complete description of symmetry
reductions for four three dimensional systems: the universal hierarchy
equation, the 3D rdDym equation, the modified Veronese web equation, and
Pavlov's equation. The result comprised more than~$30$ equations, but the
majority of them were either exactly solvable or linearized by the generalized
Legendre transformations. Nevertheless, there were~$10$ `interesting'
reductions, among which two well-known equations, i.e., the Liouville and
Gibbons-Tsarev equations,~\cite{Liou,GT}. The rest eight can be divided in two
groups by their symmetry properties: five equations admit infinite-dimensional
Lie algebras of contact symmetries (with functional parameters) and three
others possess finite-dimensional symmetry algebras. These are
\begin{equation}
  \label{eq:1}
  u_y u_{xy}-u_x u_{yy}=e^y u_{xx}
\end{equation}
(reduction of the universal hierarchy equation),
\begin{equation}
  \label{eq:2}
  u_{yy}=(u_x+ x)u_{xy}-u_y(u_{xx}+2)
\end{equation}
(reduction of the 3D rdDym equation), and
\begin{equation}
  \label{eq:3}
  u_{xx}=(x-u_y)u_{xy} + (2y+u_x)u_{yy} -u_y
\end{equation}
(reduction of the Pavlov equation)\footnote{All the reductions of the modified
  Veronese web equation were either exactly solvable or linearizable.}. These
equations are pair-wise inequaivalent (see Section~\ref{sec:conclusions}).

We deal with this three equations below and study how the integrability
properties of the initial 3D systems behave under reduction. More precisely,
we construct (Section~\ref{sec:reduction-lax-pairs}) the reductions of the
zero-curvature representations for Equations~\eqref{eq:1}--\eqref{eq:2} and
show that they result in differential coverings of the form
\begin{equation*}
  w_x=\frac{a_2w^2+a_1w+a_0}{w^2+c_1w+c_0},\qquad
  w_y=\frac{b_2w^2+b_1w+b_0}{w^2+c_1w+c_0},
\end{equation*}
where~$a_i$, $b_i$, $c_i$ are functions in~$x$, $y$, $u$, $u_x$,
and~$u_y$. These coverings are similar to the one found
in~\cite{O-S:NonHom-Sys} for the Gibbons-Tsarev equation and this resemblance,
by all means, reflects the relations between generalized Gibbons-Tsarev
equations and integrable 3D-systems~\cite{O-S:Sys-G-T}. In
Section~\ref{sec:hier-nonl-cons}, for every nonlinear covering we construct
an infinite series of conservation laws and prove their nontriviality.

We also study the behavior of the recursion operators for symmetries of
three-dimensional systems and show that these operators do not survive under
reduction (Section~\ref{sec:reduct-recurs-oper}).

In Section~\ref{sec:symm-cosymm-cons} local symmetries and cosymmetries of the
reduction equations are described. The corresponding conservation laws are
presented in the Appendix.

Throughout the text the notion of (differential) covering is understood in the
sense of~\cite{KV}.

\section{Reduction of the Lax pairs}
\label{sec:reduction-lax-pairs}

Using Lax representations of the 3D equations, whose reductions are the
equations at hand, we construct here nonlinear coverings of
Equations~\eqref{eq:1}--\eqref{eq:3}.

\subsection{Equation~\eqref{eq:1}}
\label{sec:equation-eqrefeq:1}

This equation is obtained as the reduction of the universal hierarchy
equation\footnote{To save the notation here and below, we denote by~$u$ the
  dependent and by~$x$, $y$ the dependent variables. These are \emph{not} the
  same as in the initial equation; see the details in~\cite{B-K-M-V:Sym-Red}.}
\begin{equation}
  \label{eq:4}
  u_{yy} = u_z u_{xy} - u_y u_{xz}
\end{equation}
with respect to the symmetry
\begin{equation}
    \label{uhe_symmetry}
  \phi = u_z + u_x + y u_y + u.
\end{equation}
Equivalently, this reduction may be written in the form
\begin{equation}
  \label{eq:5}
  u_{yy} = u_y u_{xx} - (u_x + u) u_{xy} + u_x u_y
\end{equation}
and Equation~\eqref{eq:1} transforms to~\eqref{eq:5} by the change of
variables~$x\mapsto y$, $y\mapsto x$, $u\mapsto -e^yu$.

Equation~\eqref{eq:4} admits the following Lax representation
\begin{equation}\label{eq:6}
  \begin{array}{ll}
    w_z &= (w u_z - u_y) w^{-2} w_x,\\[2pt]
    w_y &= u_y w^{-1} w_x.
  \end{array}
\end{equation}
The symmetry~$\phi$ can be extended to a symmetry~$\Phi=(\phi,\chi)$
of~\eqref{eq:6}, where
\begin{equation*}
  \chi=w_z + w_x + y w_y + w
\end{equation*}
and the corresponding reduction leads to the covering
\begin{equation}
  \label{eq:7}
  \begin{array}{ll}
    w_x&=-\dfrac{w^3}{w^2 - (u_x + u) w - u_y}, \\[10pt]
    w_y&=-\dfrac{u_y w^2}{w^2 - (u_x + u) w - u_y}
  \end{array}
\end{equation}
of Equation~\eqref{eq:5}.

\begin{remark}
  Equation~\eqref{eq:1} can be written in the potential form
  \begin{equation*}
   \left(\frac{u_y}{u_x}\right)_y = \left(\frac{e^y}{u_x}\right)_x,
  \end{equation*}
  the corresponding Abelian covering being
  \begin{equation}
    v_x=\frac{u_y}{u_x},\qquad v_y=\frac{e^y}{u_x}.
  \end{equation}
  Then~$v$ enjoys the equation
  \begin{equation}
    \label{eq:8}
    v_y - v_{yy} = v_yv_{xx} - v_xv_{xy},
  \end{equation}
  which also admits the rational covering
  \begin{equation*}
    \begin{array}{ll}
      w_x&=\dfrac{w v_x-x v_x+v_y}{w^2+(-2 x+v_x) w+x^2-x v_x+v_y},\\[10pt]
      w_y&=\dfrac{w v_y-x v_y}{w^2+(-2 x+v_x) w+x^2-x v_x+v_y}.
    \end{array}
  \end{equation*}
  of the same type.\QED
\end{remark}

\subsection{Equation~\eqref{eq:2}}
\label{sec:equation-eqrefeq:2}

This equation was obtained as the reduction of the 3D rdDym equation
\begin{equation}
  \label{eq:9}
  u_{ty} = u_x u_{xy} - u_y u_{xx}
\end{equation}
with respect to the symmetry
\begin{equation}
  \label{3DrdDym_symmetry}
  \phi= u_t - x u_x - u_y + 2  u.
\end{equation}
The Lax representation for Equation~\eqref{eq:9} is
\begin{equation}
  \label{eq:10}
  \begin{array}{ll}
    w_t &= (u_x + w) w_x,\\[2pt]
    w_y &= -u_y w^{-1} w_x.
  \end{array}
\end{equation}
The symmetry~$\phi$ extends to the one of~\eqref{eq:10}: $\Phi=(\phi,\chi)$,
where
\begin{equation*}
  \chi=w_t - x w_x - w_y +  u.
\end{equation*}
Reduction of the covering~\eqref{eq:10} with respect to~$\Phi$ leads to the
covering
\begin{equation}\label{eq:22}
  \begin{array}{ll}
    w_x&=-\dfrac{w^2}{w^2 + (u_x - x) w + u_y},\\[10pt]
    w_y&= \dfrac{u_y w}{w^2 + (u_x - x) w + u_y}.
  \end{array}
\end{equation}
over Equation~\eqref{eq:2}.

\subsection{Equation~\eqref{eq:3}}
\label{sec:equation-eqrefeq:3}

Finally, Equation~\eqref{eq:3} is the reduction of the Pavlov equation
\begin{equation}
  \label{eq:11}
  u_{yy} = u_{tx} + u_y u_{xx} - u_x u_{xy}
\end{equation}
with respect to the symmetry
\begin{equation}\label{Pavlov_symmetry}
  \phi = u_t - 2 x u_x - y u_y + 3 u.
\end{equation}
The Pavlov equation possesses the Lax pair
\begin{equation}
  \label{eq:12}
  \begin{array}{ll}
    w_t& = (w^2 - w u_x - u_y) w_x,\\[2pt]
    w_y& = (w - u_x) w_x.
  \end{array}
\end{equation}
The symmetry~$\phi$ lifts to the symmetry~$\Phi=(\phi,\chi)$ of~\eqref{eq:12},
where
\begin{equation*}
  \chi=w_t - 2 x w_x - y w_y + w.
\end{equation*}
Reduction of the covering~\eqref{eq:12} with respect to this symmetry results
in the nonlinear covering
\begin{equation}
  \label{eq:13}
  \begin{array}{ll}
    w_x&=-\dfrac{w (w - u_y)}{w^2 - (u_y + x) w + x u_y - u_x - 2 y},\\[10pt]
    w_y&=-\dfrac{w}{w^2 - (u_y + x) w + x u_y - u_x - 2 y}
  \end{array}
\end{equation}
of Equation~\eqref{eq:3}.

\begin{remark}
  Equation~\eqref{eq:3} has a close relative. Namely, if we accomplish
  reduction of the Pavlov equation using another symmetry
  \begin{equation*}
    \phi' = u_t - y u_x + 2  x
  \end{equation*}
  the resulting equation will be
  \begin{equation}
    \label{eq:14}
    u_{yy} = (u_y +  y) u_{xx} - u_x u_{xy} - 2.
  \end{equation}
  The symmetry~$\phi'$ can also be lifted to~\eqref{eq:12}
  by~$\Phi'=(\phi',\chi')$, where
  \begin{equation*}
    \chi'=w_t - y w_x + 1,
  \end{equation*}
  and the reduction of~\eqref{eq:12} will be
  \begin{equation}
    \label{eq:15}
    \begin{array}{ll}
      w_x&=-\dfrac{1}{w^2 - u_x w - u_y - y},\\[10pt]
      w_y&=-\dfrac{w - u_x}{w^2 - u_x w - u_y - y}.
    \end{array}
  \end{equation}
  By the change of variables~$u\mapsto u-y^2/2$, Equation~\eqref{eq:14}
  transforms to the Gibbons-Tsarev equation
  \begin{equation*}
    u_{yy} = u_y u_{xx} - u_x u_{xy} - 1,
  \end{equation*}
  while~\eqref{eq:13} becomes
  \begin{equation*}
    \begin{array}{ll}
      w_x&=-\dfrac{1}{w^2 - u_x w - u_y},\\[10pt]
      w_y&=-\dfrac{w - u_x}{w^2 - u_x w - u_y},
    \end{array}
  \end{equation*}
  cf.~\cite{O-S:NonHom-Sys}.\QED
\end{remark}

\section{Local symmetries and cosymmetries of the reduction equations}
\label{sec:symm-cosymm-cons}

We present here computational results on classical symmetries and cosymmetries
of Equations~\eqref{eq:1}--\eqref{eq:3}, i.e., solutions of the equations
\begin{equation*}
  \ell_{\mathcal{E}}(\phi)=0
\end{equation*}
and
\begin{equation*}
  \ell_{\mathcal{E}}^*(\psi)=0,
\end{equation*}
where~$\ell_{\mathcal{E}}$ is the linearization of the equation at hand
and~$\ell_{\mathcal{E}}^*$ is its formally adjoint and~$\phi$ and~$\psi$
depend on~$x$, $y$, $u$, $u_x$, $u_y$ (see, e.g.,~\cite{KLV}). The
conservation laws corresponding to classical cosymmetries are presented in the
Appendix below. The spaces of solutions are denoted by~$\sym(\mathcal{E})$
and~$\cosym(\mathcal{E})$, respectively.

All the equations under consideration happen to possess a scaling symmetry and
thus admit weights (which we denote by~$\abs{\cdot\,}$) with respect to which
they become homogeneous.

\subsection{Equation~\eqref{eq:1}}
\label{sec:equation-eqrefeq:1-1}

We consider this equation in the form~\eqref{eq:5}, i.e.,
\begin{equation*}
  u_{yy} = u_y u_{xx} - (u_x + u) u_{xy} + u_x u_y.
\end{equation*}
The weights are
\begin{equation*}
  \abs{x}=0,\quad \abs{y}=1, \quad\abs{u}=-1, \quad\abs{u_x}=-1,
  \quad\abs{u_y}=-2.
\end{equation*}

\subsubsection*{Symmetries}
\label{sec:symmetries}

The defining equation for symmetries is\footnote{Here and below~$D_x$
  and~$D_y$ denote the total derivatives with respect to~$x$ and~$y$.}
\begin{equation*}
  D_y^2(\phi) = u_yD_x^2(\phi) - (u_x+u)D_xD_y(\phi) + (u_y-u_{xy})D_x(\phi) +
  (u_{xx}+u_x)D_y(\phi) - u_{xy}\phi.
\end{equation*}
The space~$\sym(\mathcal{E})$ spans the symmetries
\begin{equation*}
  \phi_{-1}=u_y,\quad \phi_0= yu_y+u,\quad \phi'_0=u_x,\quad \phi_1 = e^{-x},  
\end{equation*}
where the subscripts coincide with the weights\footnote{To a symmetry~$\phi$ we
  assign the weight of the corresponding evolutionary vector
  field~$\Ev_\phi$.}.

\subsubsection*{Cosymmetries}
\label{sec:cosymmetries}

The defining equation for cosymmetries of Equation~\eqref{eq:1} is
\begin{equation*}
  D_y^2(\psi)=u_yD_x^2(\psi) -(u_x+u)D_xD_y(\psi) + 2(u_{xy}+u_y)D_x(\psi) -
  2(u_{xx}+u_x)D_y(\psi) - 3u_{xy}\psi.
\end{equation*}
The space~$\cosym(\mathcal{E})$ is $6$-dimensional and spans the following
cosymmetries:
\begin{equation*}
  \psi_{-3} = e^{4 x} (3 u_x^2+8 u^2+10 u u_x+2 u_y),\quad
  \psi_{-2} = e^{3 x} (3 u+2 u_x), \quad
  \psi_{-1} = e^{2 x}
\end{equation*}
and
\begin{gather*}
  \psi_3 = \frac{1}{u_y^2},\qquad \psi_4 =\frac{2 u_x-y u_y+2 u}{u_y^3},\\
  \psi_5 =\frac{-4 u_x y u_y+6 u u_x+3 u_x^2-4 y u u_y+3 u^2+2 u_y+y^2
    u_y^2}{u_y^4},
\end{gather*}
where superscript coincides with the weight\footnote{To every cosymmetry we
  assign the weight of the corresponding variational form,
  see~\cite{K-V:GomJetSpaces}}.

\subsection{Equation~\eqref{eq:2}}
\label{sec:equation-eqrefeq:2-1}

The weights are
\begin{equation*}
  \abs{x}=1,\quad \abs{y}=0,\quad\abs{u}=2,\quad\abs{u_x}=1,\quad\abs{u_y}=2.
\end{equation*}

\subsubsection*{Symmetries}
\label{sec:symmetries-1}

The linearized equation is
\begin{equation*}
  D_y^2(\phi)=(u_x+x)D_xD_y(\phi) -u_yD_x^2(\phi) + u_{xy}D_x(\phi) -
  (u_{xx}+2)D_y(\phi). 
\end{equation*}
The space~$\sym(\mathcal{E})$ is generated by the symmetries
\begin{equation*}
  \phi_{-2}=1,\quad \phi_{-1}=u_x+x,\quad \phi_0=u-\frac{1}{2}xu_x,\quad
  \phi'_0=u_y. 
\end{equation*}

\subsubsection*{Cosymmetries}
\label{sec:cosymmetries-1}

The defining equation for cosymmetries reads
\begin{equation*}
  D_y^2(\psi)=(u_x+x)D_xD_y(\psi) - u_yD_x^2(\psi) - 2u_{xy}D_x(\psi) +
  (2u_{xx}+3)D_y(\psi). 
\end{equation*}
The space~$\cosym(\mathcal{E})$ is generated by the cosymmetries
\begin{align*}
  \psi_{-3}&=\frac{e^{-2 y} (u_x+x)}{u_y^3},&&\psi_2=1,\\
  \psi_{-2}&=\frac{e^{-y}}{u_y^2},&&\psi_3=u_x+2 x.
\end{align*}

\subsection{Equation~\eqref{eq:3}}
\label{sec:equation-eqrefeq:3-1}

The weights of variables are
\begin{equation*}
  \abs{x}=1,\quad \abs{y}=2,\quad \abs{u}=3,\quad \abs{u_x}= 2,\quad
  \abs{u_y}=1. 
\end{equation*}
in this case.

\subsubsection*{Symmetries}
\label{sec:symmetries-2}

The symmetries are defined by the equation
\begin{equation*}
  D_x^2(\phi)=(x-u_y)D_xD_y(\phi) +(2y+u_x)D_y^2(\phi)-D_y(\phi)
\end{equation*}
and the space~$\sym(\mathcal{E})$ spans the symmetries\begin{align*}
  \phi_0&=-\frac{1}{3} x u_x -\frac{2}{3} y u_y + u, &&
  \phi_{-1}= u_x - x u_y + y-\frac{1}{2} x^2, \\
  \phi_{-2}&= u_y +2 x, && \phi_{-3}=1.
\end{align*}

\subsubsection*{Cosymmetries}
\label{sec:cosymmetries-2}

The defining equation for cosymmetries is of the form
\begin{equation*}
  D_x^2(\psi) = (x-u_y)D_xD_y(\psi) + (2y+u_x)D_y^2 - u_{yy}D_x +3(2-u_{xy})D_y.
\end{equation*}
The space~$\cosym(\mathcal{E})$ is $6$-dimensional and spans the elements 
\begin{align*}
  \psi_7&=\frac{54}{5} x u_x u_y + \frac{164}{5} x u_y y + \frac{256}{5} x^2 y
  + 2 x u + \frac{4}{5} u u_y+\frac{12}{5} u_y^2 u_x + 4 y u_x + \frac{36}{5}
  u_y^2 y\\
  & + \frac{82}{5} x^2 u_x + \frac{512}{15} x^3 u_y+ \frac{32}{5} x u_y^3 +
  \frac{96}{5} x^2 u_y^2 + \frac{32}{5} y^2 +
  \frac{512}{15} x^4 + \frac{3}{5} u_x^2+u_y^4, \\
  \psi_6&=\frac{49}{4} x y + 4 x u_x + \frac{3}{2} u_y u_x + \frac{9}{2} u_y y
  + \frac{49}{4} x^2 u_y + \frac{21}{4} x u_y^2 + \frac{343}{24} x^3 +
  \frac{1}{4}  u+u_y^3, \\
  \psi_5&=4 x u_y+6 x^2+2 y+\frac{2}{3} u_x+u_y^2, \\
  \psi_4&=\frac{5}{2} x+u_y, \\
  \psi_3&=1,\\
  \psi_{-1}&=\frac{1}{(-x u_y+u_x+2 y)^2}.
\end{align*}

\section{Hierarchies of nonlocal conservation laws}
\label{sec:hier-nonl-cons}

Using the nonlinear coverings presented in
Section~\ref{sec:reduction-lax-pairs} we construct here infinite hierarchies
of nonlocal conservation laws for Equations~\eqref{eq:1}--\eqref{eq:1}.

\subsection{A general construction}
\label{sec:general-construction}

The initial step of the construction is the so-called \emph{Pavlov
  reversing},~\cite{P-J-Y:Int-M-S-H} (see~\cite{I.K:Nat-Constr} for the
invariant geometrical interpretation). Let~$\mathcal{E}$ be an equation in two
independent variables~$x$ and~$y$ and unknown function~$u$ and
\begin{equation*}
  w_x=X(x,y,[u],w),\qquad w_y=Y(x,y,[u],w)
\end{equation*}
be a differential covering over~$\mathcal{E}$, where~$[u]$ denotes~$u$ itself
and a collection of its derivatives up to some finite order. Then the
system
\begin{equation}
  \label{eq:16}
  \psi_x=-X(x,y,[u],\lambda)\psi_\lambda,\qquad
  \psi_y=-Y(x,y,[u],\lambda)\psi_\lambda 
\end{equation}
is also compatible modulo~$\mathcal{E}$ (thus, the nonlocal variable~$w$ turns
into a formal parameter in the new setting).

Assume now that
\begin{align*}
  X&=X_{-1}\lambda+X_0+\frac{X_1}{\lambda}+\dots+\frac{X_i}{\lambda^i}+\dots,\\
  Y&=Y_{-1}\lambda+Y_0+\frac{Y_1}{\lambda}+\dots+\frac{Y_i}{\lambda^i}+\dots,
\end{align*}
where~$X_i$, $Y_i$, $i\geq-1$, are functions in~$x$, $y$ and~$[u]$, and also
expand~$\psi$ in formal Laurent series
\begin{equation*}
  \psi=\psi_{-1}\lambda+\psi_0+\frac{\psi_1}{\lambda}+\dots+
  \frac{\psi_i}{\lambda^i}+\dots
\end{equation*}
Then~\eqref{eq:16} implies
\begin{equation*}
  \psi_{i,x}=-\sum_{j+k=i+1}kX_j\psi_k,\qquad\psi_{i,y}=-\sum_{j+k=i+1}kY_j\psi_k,
\end{equation*}
or  
\begin{align*}
  &\psi_{-1,x}=-X_{-1}\psi_{-1},&&\psi_{-1,y}=-Y_{-1}\psi_{-1};\\
  &\psi_{0,x}=-X_0\psi_{-1},&&\psi_{0,y}=-Y_0\psi_{-1};\\
  &\psi_{1,x}=X_{-1}-X_1\psi_{-1},&&\psi_{1,y}=Y_{-1}-Y_1\psi_{-1};\\
  &\psi_{2,x}=2X_{-1}\psi_2+X_0\psi_1-X_2\psi_{-1},&&
  \psi_{2,y}=2Y_{-1}\psi_2+Y_0\psi_1-Y_2\psi_{-1};\\
  &\dots&&\dots
\end{align*}
and
\begin{align*}
  \psi_{k,x}&=kX_{-1}\psi_k+(k-1)X_0\psi_{i-1}+\dots+X_{k-2}\psi_1-X_k\psi_{-1},\\
  \psi_{k,y}&=kY_{-1}\psi_k+(k-1)Y_0\psi_{i-1}+\dots+Y_{k-2}\psi_1-Y_k\psi_{-1}\\
\end{align*}
for all~$k>2$.

In general, this system defines an infinite-dimensional non-Abelian covering
(which may be trivial generally) over the base equation~$\mathcal{E}$, but in
the particular case~$X_{-1}=Y_{-1}=0$ the covering becomes Abelian, i.e.,
transforms to an infinite series of (nonlocal) conservation laws. Indeed, the
first pair of equations reads
\begin{equation*}
  \psi_{-1,x}=0,\qquad\psi_{-1,y}=0
\end{equation*}
in this case and without loss of generality we may set~$\psi_{-1}=1$. The rest
equations read
\begin{align*}
  &\psi_{0,x}=-X_0,&&\psi_{0,y}=-Y_0;\\
  &\psi_{1,x}=-X_1,&&\psi_{1,y}=-Y_1;\\
  &\psi_{2,x}=X_0\psi_1-X_2,&&  \psi_{2,y}=Y_0\psi_1-Y_2;\\
  &\psi_{3,x}=2X_0\psi_2+X_1\psi_1-X_3,&&\psi_{3,x}=2Y_0\psi_2+Y_1\psi_1-Y_3;\\
  &\dots&&\dots
\end{align*}
and
\begin{equation}\label{eq:21}
  \begin{array}{l}
    \psi_{k,x}=(k-1)X_0\psi_{k-1}+(k-2)X_1\psi_{k-2}+\dots+X_{k-2}\psi_1-X_k,\\
    \psi_{k,y}=(k-1)Y_0\psi_{k-1}+(k-2)Y_1\psi_{k-2}+\dots+Y_{k-2}\psi_1-Y_k
\end{array}
\end{equation}
for all~$k>3$.

\begin{remark}
  The first two pairs of equations define local conservation laws (probably,
  trivial) and the potential $\psi_0$ does not enter the other equations. This
  means that the obtained covering is the Whitney product of the one-dimensional
  Abelian covering~$\tau_0$ associated to~$\psi_0$ and the
  infinite-dimensional~$\tau_*$ related to~$\psi_1$, $\psi_2,\dots$ We shall
  deal with~$\tau_*$ below.\QED
\end{remark}

We now confine ourselves to the case
\begin{equation}\label{eq:18}
  X=\frac{a_2w^2+a_1w+a_0}{w^2+c_1w+c_0},\qquad
  Y=\frac{b_2w^2+b_1w+b_0}{w^2+c_1w+c_0},
\end{equation}
where~$a_i$, $b_i$, and~$c_i$ are functions in~$x$, $y$, and~$[u]$, and deduce
the needed Laurent expansions. One has
\begin{multline*}
  \frac{a_2\lambda^2+a_1\lambda+a_0}{\lambda^2+c_1\lambda+c_0} =
  \left(a_2+\frac{a_1}{\lambda}+\frac{a_0}{\lambda^2}\right)\cdot
  \left(\frac{1}{1+\frac{c_1\lambda+c_0}{\lambda^2}}\right)\\ =
  \left(a_2+\frac{a_1}{\lambda}+\frac{a_0}{\lambda^2}\right) \cdot
  \sum_{i\geq0} \left(-\frac{c_1\lambda+c_0}{\lambda^2}\right)^i.
\end{multline*}
Let us present temporally the second factor in the form
\begin{equation*}
  \sum_{i\geq0} \left(-\frac{c_1\lambda+c_0}{\lambda^2}\right)^i = \sum_{i\geq
  0}\frac{d_i}{\lambda^i}.
\end{equation*}
Then
\begin{multline*}
  \frac{a_2\lambda^2+a_1\lambda+a_0}{\lambda^2+c_1\lambda+c_0} =
  \left(a_2+\frac{a_1}{\lambda}+\frac{a_0}{\lambda^2}\right) \cdot \sum_{i\geq
    0}\frac{d_i}{\lambda^i} \\
  = a_2d_0 + \frac{a_2d_1+a_1d_0}{\lambda} +
  \frac{a_2d_2+a_1d_1+a_0d_0}{\lambda^2} + \dots +
  \frac{a_2d_i+a_1d_{i-1}+a_0d_{i-2}}{\lambda^i} +\dots
\end{multline*}
Compute the coefficients~$d_i$ now. One has
\begin{equation*}
  \left(-\frac{c_1\lambda+c_0}{\lambda^2}\right)^i =
  (-1)^i\sum_{j=0}^i\binom{i}{j}\frac{c_1^jc_0^{i-j}}{\lambda^{2i-j}},
\end{equation*}
from where it follows that
\begin{equation*}
  d_0=1,\qquad d_1=-c_1
\end{equation*}
and
\begin{equation}\label{eq:19}
  d_i=
  \begin{cases}\displaystyle
    \sum_{j=0}^k(-1)^{k-j}\binom{k+j}{2j}c_0^{k-j}c_1^{2j}&\text{if }
    i=2k,\\[4pt]
    \displaystyle
    \sum_{j=0}^k(-1)^{k-j+1}\binom{k+j+1}{2j+1}c_0^{k-j}c_1^{2j+1}&\text{if } i=2k+1
  \end{cases}
\end{equation}
for~$i>1$, 
Or, in shorter notation
\begin{equation}
  \label{eq:17}
  d_i=\sum_{j=0}^{[i/2]}(-1)^{[i/2]-j+p(i)}
  \binom{[i/2]+j+p(i)}{2j+p(i)}c_0^{[i/2]-j}c_1^{2j+p(i)}, 
\end{equation}
where~$p(i)=i\bmod 2$ is the parity of~$i$ and~$[k/2]$ is the integer part.

Gathering together the results of the above computations, one obtains that in
the case of coverings~\eqref{eq:18} we have~$X_{-1}=Y_{-1}=0$, while other
coefficients are
\begin{align*}
  &X_0=a_2,&&Y_0=b_2;\\
  &X_1=a_1-a_2c_1,&&Y_1=b_1-b_2c_1;\\
  &X_2=a_0-a_1c_1+a_2(c_1^2-c_0),&&Y_2=b_0-b_1c_1+b_2(c_1^2-c_0);\\
  &\dots&&\dots\\
  &X_i=a_0d_{i-2}+a_1d_{i-1}+a_2d_i,&&Y_i=b_0d_{i-2}+b_1d_{i-1}+b_2d_i;\\
  &\dots&&\dots,
\end{align*}
where the functions~$d_i$ are given by~\eqref{eq:19}.

Let us now show how these general constructions look like in the particular
cases of the equations under consideration.

\subsection{Equation~\eqref{eq:1}}
\label{sec:equation-eqrefeq:1-3}

Note first that the covering~\eqref{eq:7} is not of the
form~\eqref{eq:18}. Nevertheless, it can be transformed to the needed form by
the gauge transformation~$w\mapsto we^{-x}$. Then~\eqref{eq:7} acquires the
form
\begin{equation*}
  w_x=\frac{(u_x+u)e^xw^2-u_ye^{2x}w}{w^2-(u_x+u)e^xw-u_ye^{2x}},\qquad
  w_y=-\frac{u_ye^xw^2}{w^2-(u_x+u)e^xw-u_ye^{2x}}.
\end{equation*}
We have~$\abs{w}=-1$.

Thus,
\begin{align*}
  &a_0=0,&&a_1=-u_ye^{2x}&&a_2=(u_x+u)e^x,\\
  &b_0=0,&&b_1=0,&&b_2=-u_ye^x,\\
  &c_0=-u_ye^{2x},&&c_1=-(u_x+u)e^x.
\end{align*}
Let us compute the coefficients~$d_i$. By~\eqref{eq:19}, we have
\begin{multline*}
  d_{2k}=\sum_{j=0}^k(-1)^{k-j}\binom{k+j}{2j}\left(-u_ye^{2x}\right)^{k-j}
  \left(-(u_x+u)e^{x}\right)^{2j}\\ =
  e^{2kx}\sum_{j=0}^k\binom{k+j}{2j}u_y^{k-j}(u_x+u)^{2j} 
\end{multline*}
and
\begin{multline*}
  d_{2k+1}=\sum_{j=0}^k(-1)^{k-j+1}\binom{k+j+1}{2j+1}\left(-u_ye^{2x}\right)^{k-j}
  \left(-(u_x+u)e^x\right)^{2j+1}\\ = e^{(2k+1)x}\sum_{j=0}^k\binom{k+j+1}{2j+1}
  u_y^{k-j}(u_x+u)^{2j+1},
\end{multline*}
or
\begin{equation}\label{eq:20}
  d_i=e^{ix}\sum_{j=0}^{[i/2]}\binom{[i/2]+j+p(i)}{2j+p(i)}
  u_y^{[i/2]-j}(u_x+u)^{2j+p(i)}.
\end{equation}
Hence,
\begin{align*}
  &X_0=(u_x+u)e^x,&&Y_0=-u_ye^x;\\
  &X_1=\left((u_x+u)^2-u_y\right)e^{2x},&&Y_1=(u_x+u)u_ye^{2x}
\end{align*}
and
\begin{align*}
  X_i&=e^{(i+1)x}\left((u_x+u)^{i+1} +
    \sum_{j=1}^{[(i+1)/2]}\left(\binom{i-j}{i-2j} - \binom{i-j}{i-2j+1}\right)
      u_y^j(u_x+u)^{i-2j+1}\right),\\
  Y_i&=-e^{(i+1)x}\sum_{j=0}^{[i/2]}\binom{[i/2]+j+p(i)}{2j+p(i)}
  u_y^{[i/2]-j+1}(u_x+u)^{2j+p(i)}
\end{align*}
for~$i>1$ (we assume~$\binom{\alpha}{\beta}=0$ for~$\beta<0$). Obviously,
\begin{equation*}
  \abs{X_i}=-i-1,\qquad\abs{Y_i}=-i-2.
\end{equation*}
The functions~$X_i$, $Y_i$ define, by Equations~\eqref{eq:21}, the infinite
number of nonlocal variables~$\psi_i$ for Equation~\eqref{eq:1} with
\begin{equation*}
  \abs{\psi_i}=-i-1.
\end{equation*}
The corresponding conservation laws have the same weights and the first three
of them coincide (up to equivalence) with the local conservation
laws~$\omega_{-2}$, $\omega_{-3}$, $\omega_{-4}$ described in
Section~\ref{sec:equation-eqrefeq:1-1}. The first essentially nonlocal one is
associated to~$\psi_3$.

\subsection{Equation~\eqref{eq:2}}
\label{sec:equation-eqrefeq:2-3}

Due to Equations~\eqref{eq:22}, one has
\begin{align*}
  &a_0=0,&&a_1=0&&a_2=-1,\\
  &b_0=0,&&b_1=u_y,&&b_2=0,\\
  &c_0=u_y,&&c_1=u_x-x.
\end{align*}
Hence,
\begin{align*}
  X_0&=-1,&&Y_0=0;\\
  X_1&=u_x-x,&&Y_1=u_y;\\
  X_2&=-(u_x-x)^2+u_y,&&Y_2=-u_y(u_x-x)
\end{align*}
and
\begin{align*}
  X_i&=-d_i=\sum_{j=0}^{[i/2]}(-1)^{[i/2]-j+p(i)+1}
  \binom{[i/2]+j+p(i)}{2j+p(i)}u_y^{[i/2]-j}(u_x-x)^{2j+p(i)},\\
  Y_i&=u_yd_{i-1}=\sum_{j=0}^{[(i-1)/2]}(-1)^{[(i-1)/2]-j+p(i-1)}\times\\
  &\times\binom{[(i-1)/2]+j+p(i-1)}{2j+p(i-1)}u_y^{[(i-1)/2]-j+1}(u_x-x)^{2j+p(i-1)}
\end{align*}
for~$i>2$.
Consequently,
\begin{align*}
  \psi_{0,x}&=-X_0=1,&&\psi_{0,y}=-Y_0=0;\\
  \psi_{1,x}&=-X_1=-u_x+x,&&\psi_{1,y}=-Y_1=-u_y
\end{align*}
and one may set
\begin{equation*}
  \psi_0=x,\qquad\psi_1=-u+\frac{x^2}{2},
\end{equation*}
while
\begin{equation*}
  \psi_{2,x}=(u_x-x)^2+u_y+u-\frac{x^2}{2},\quad\psi_{2,y}=(u_x-x)u_y
\end{equation*}
and for~$i>2$
\begin{align*}
  \psi_{i,x}&=-(i-1)\psi_{i-1} + (i-2)X_1\psi_{i-2} + \dots + X_{i-3}\psi_2 +
  \left(\frac{x^2}{2}-u\right)X_{i-2} - X_i,\\
  \psi_{i,y}&=(i-2)Y_1\psi_{i-2} + \dots + Y_{i-3}\psi_2 +
  \left(\frac{x^2}{2}-u\right)Y_{i-2} - Y_i,
\end{align*}
where~$X_k$, $Y_k$ are given by the above formulas.

One has
\begin{equation*}
  \abs{X_i}=i,\quad\abs{Y_i}=i+1,\quad\abs{\psi_i}=i+1.
\end{equation*}
The conservation law corresponding to~$\psi_i$ is of the weight~$i+1$ and the
first two ones, up to equivalence coincide with those described in
Section~\ref{sec:equation-eqrefeq:2-1}, while all the others are essentially
nonlocal.

\subsection{Equation~\eqref{eq:3}}
\label{sec:equation-eqrefeq:3-3}

By Equation~\eqref{eq:13}, we have
\begin{align*}
  &a_0=0,&&a_1=u_y&&a_2=-1,\\
  &b_0=0,&&b_1=-1,&&b_2=0,\\
  &c_0=xu_y-u_x-2y,&&c_1=-(u_y+x).
\end{align*}
Consequently,
\begin{align*}
  X_0&=-1,&&Y_0=0;\\
  X_1&=-x,&&Y_1=-1;\\
  X_2&=-u_x-x^2-2y,&&Y_2=-u_y-x
  \intertext{and}
  X_i&=u_yd_{i-1}-d_i,&&Y_i=-d_{i-1}
\end{align*}
for~$i>2$, where
\begin{equation*}
  d_i=\sum_{j=0}^{[i/2]}(-1)^{[i/2]-j}
  \binom{[i/2]+j+p(i)}{2j+p(i)}(xu_y-u_x-2y)^{[i/2]-j}(u_y+x)^{2j+p(i)}.
\end{equation*}
One has
\begin{equation*}
  \abs{X_i}=i,\qquad\abs{Y_i}=i-1.
\end{equation*}

Thus we have
\begin{align*}
  \psi_{1,x}&=x,&&\psi_{1,y}=1;\\
  \psi_{2,x}&=u_x+\frac{x^2}{2}+y,&&\psi_{2,y}=u_y+x
\end{align*}
and we may set
\begin{equation*}
  \psi_1=\frac{x^2}{2}+y,\qquad\psi_{2}=u+xy+\frac{x^3}{6}.
\end{equation*}
Then the other potentials are defined by
\begin{align*}
  \psi_{i,x}&=-(i-1)\psi_{i-1}-(i-2)\psi_{i-2}(i-3)X_2\psi_{i-3}+ \dots \\
  &\ \dots+
  3X_{i-4}\psi_3 + \left(2u+2xy+\frac{x^3}{3}\right)X_{i-3} +
  \left(\frac{x^2}{2} +y\right)X_{i-2}-X_i,\\
  \psi_{i,y}&=-(i-2)\psi_{i-2}(i-3)Y_2\psi_{i-3}+ \dots \\
  &\ \dots+
  3Y_{i-4}\psi_3 + \left(2u+2xy+\frac{x^3}{3}\right)Y_{i-3} +
  \left(\frac{x^2}{2} +y\right)Y_{i-2}-Y_i,
\end{align*}
$i>2$. We have
\begin{equation*}
  \abs{\psi_i}=i+1.
\end{equation*}

The conservation laws associated with~$\psi_3,\dots,\psi_7$ are equivalent
to~$\omega_4,\dots,\omega_8$ introduced in
Section~\ref{sec:equation-eqrefeq:3-1}. The first essentially nonlocal
conservation law corresponds to~$\psi_8$.

\subsection{Proof of nontriviality}
\label{sec:proof-nontriviality}

We shall now prove that the above constructed conservation laws are
nontrivial. To this end, introduce the notation~$\mathcal{E}_\alpha$,
$\alpha=1$, $2$, $3$, for Equations~\eqref{eq:1}, \eqref{eq:2}
and~\eqref{eq:3}, respectively, and
\begin{equation*}
  \tau_{i,\alpha}\colon\mathcal{E}_{i,\alpha}\to\mathcal{E}_{\alpha}
\end{equation*}
for the coverings defined by the nonlocal
variables~$\psi_\alpha,\dots,\psi_i$. Let
\begin{equation*}
  D_x^{i,\alpha},\qquad  D_y^{i,\alpha}
\end{equation*}
be the total derivatives on~$\mathcal{E}_{i,\alpha}$.

\begin{proposition}\label{prop-1}
  For all~$i\geq\alpha$\textup{,} the only solutions of the system
  \begin{equation}\label{eq:23}
    D_x^{i,\alpha}(f)=0,\qquad  D_y^{i,\alpha}(f)=0
  \end{equation}
  are constants.
\end{proposition}

\begin{proof}
  Let us present the total derivatives in the form
  \begin{equation*}
    D_x^{i,\alpha}=D_x^{\alpha}+X^{i,\alpha},\quad
    D_y^{i,\alpha}=D_y^{\alpha}+Y^{i,\alpha}, 
  \end{equation*}
  where~$D_x^{\alpha}$, $D_y^{\alpha}$ are the total derivatives
  on~$\mathcal{E}_\alpha$ and~$X^{i,\alpha}$, $Y^{i,\alpha}$ are the `nonlocal
  tails':
  \begin{equation*}
    X^{i,\alpha}=\sum_{j=\alpha}^iX_j^{i,\alpha}\frac{\partial}{\partial\psi_j},
    \qquad 
    Y^{i,\alpha}=\sum_{j=\alpha}^iY_j^{i,\alpha}\frac{\partial}{\partial\psi_j},
  \end{equation*}
  $X_j^{i,\alpha}$, $Y_j^{i,\alpha}$ being the right-hand sides of the defining
  equations~\eqref{eq:21} for the potentials~$\psi$.

  From the constructions of
  Sections~\ref{sec:equation-eqrefeq:1-3}--\ref{sec:equation-eqrefeq:3-3} one
  readily sees that the quantities~$X_j^{i,\alpha}$ and~$Y_j^{i,\alpha}$ are
  polynomials in~$u_x$ and~$u_y$ and, moreover,
  \begin{align*}
    X^{i,1}&=\pm e^{(i+1)x}u_x^{i+1}\frac{\partial}{\partial\psi_i}+o,&&
    Y^{i,1}=\pm e^{(i+1)x}u_x^{i}u_y\frac{\partial}{\partial\psi_i}+o;\\ 
    X^{i,2}&=\pm u_x^i\frac{\partial}{\partial\psi_i}+o;&&
    Y^{i,2}=\pm u_x^{i-1}\frac{\partial}{\partial\psi_i}+o;\\
    X^{i,3}&=\pm u_y^{i-2}u_x\frac{\partial}{\partial\psi_i}+o,&&
    Y^{i,3}=\pm u_y^{i-1}\frac{\partial}{\partial\psi_i}+o,
  \end{align*}
  where~$o$ denotes terms of lower degree.

  Now, the proof goes by induction. For small~$i$'s the result follows from
  the fact that the cosymmetries corresponding to the local conservation laws
  do not vanish and these conservation laws are of different weights. Assume
  now that the statement is valid for all~$k<i$ and consider
  Equation~\eqref{eq:23}. Then from the above estimates it follows
  that~$\partial f/\partial\psi_i=0$.
\end{proof}

Evidently, nontriviality of the constructed conservation laws is a direct
consequence of the Proposition~\ref{prop-1}.

\section{On reductions of the recursion operators}
\label{sec:reduct-recurs-oper}

We show here that symmetry reductions of Equations~\eqref{eq:4}, \eqref{eq:9},
and~\eqref{eq:11} are incompatible with their recursion operators and thus
the latter are not inherited by Equations~\eqref{eq:1}, \eqref{eq:2},
and~\eqref{eq:3}, respectively.

\subsection{A general construction}
\label{sec:general-construction-1}

We treat here recursion operators for symmetries as B\"{a}cklund
transformations of the tangent coverings, cf.~\cite{Mar:Another}. More
precisely, let~$\mathcal{E}$ be a differential equation given by the system
\begin{equation*}
  \mathcal{E}=\{F=0\},\quad F=(F^1(x,y,[u]),\dots,F^s(x,y,[u])),
\end{equation*}
$F^j$ being functions on some jet space,~\cite{KLV}. Here, as above,~$[u]$
denotes the collection of~$u$ and its derivatives. The tangent
covering~$\tc=\tc_{\mathcal{E}}\colon\mathcal{T}\mathcal{E}\to\mathcal{E}$ is
the projection~$(x,y,[u],[q])\mapsto(x,y,[u])$ of the system
\begin{equation*}
  \mathcal{T}\mathcal{E}=\{F(x,y,[u])=0,\ \ell_F(x,y,[u],[q])=0\}
\end{equation*}
to~$\mathcal{E}$. The characteristic property of~$\tc$ is that its sections
that preserve the Cartan (higher contact) distribution are identified with
symmetries of~$\mathcal{E}$.

A B\"{a}cklund transformation between equations~$\mathcal{E}_1$
and~$\mathcal{E}_2$ is a diagram
\begin{equation*}
  \xymatrix{
    &\mathcal{B}\ar[dr]^{\tau_2}\ar[dl]_{\tau_1}&\\
    \mathcal{E}_1&&\mathcal{E}_2\rlap{,}
  }
\end{equation*}
where~$\tau_1$ and~$\tau_2$ are coverings. It relates solutions
of~$\mathcal{E}_1$ and~$\mathcal{E}_2$ to each other. A recursion operator
between symmetries of~$\mathcal{E}_1$ and~$\mathcal{E}_2$ is a B\"{a}cklund
transformation of the form
\begin{equation*}
  \xymatrix{
    &\mathcal{T}\mathcal{E}_1\ar[r]^{\tc_{\mathcal{E}_1}}&\mathcal{E}_1\\
    \mathcal{R}\ar[ur]^{\tau_1}\ar[dr]_{\tau_2}&&\\
    &\mathcal{T}\mathcal{E}_2\ar[r]^{\tc_{\mathcal{E}_2}}&\mathcal{E}_2\rlap{.}
  }
\end{equation*}
In particular, if~$\mathcal{E}_1=\mathcal{E}_2={\mathcal{E}}$ it relates
symmetries of~$\mathcal{E}$ to each other. Then~$\mathcal{R}$ may be
considered as an equation
\begin{equation*}
  \mathcal{R}\subset
  \mathcal{T}\mathcal{E}\otimes_{\mathcal{E}}\mathcal{T}\mathcal{E} 
\end{equation*}
in the Whitney product of~$\tc_{\mathcal{E}}$ with itself.

Any symmetry~$\phi=\phi(x,y,[u])$ of~$\mathcal{E}$ admits a natural
lift~$\Phi=(\phi,\phi')$ to~$\mathcal{T}\mathcal{E}$. To this end, it suffices
to set
\begin{equation*}
  \phi'=\frac{\partial\phi}{\partial u}q+ \dots +\frac{\partial\phi}{\partial
    u_\sigma}q_\sigma+ \dots
\end{equation*}

Choose a symmetry~$\phi$ of~$\mathcal{E}$ and denote
by~$\rd_\phi\colon\mathcal{E}\to\mathcal{E}_\phi$ the corresponding reduction
map. Then the diagram
\begin{equation*}
  \xymatrixcolsep{3pc}
  \xymatrix{
    \mathcal{T}\mathcal{E}\ar[r]^{\tc_{\mathcal{E}}}\ar[d]_{\rd_\Phi}&
    \mathcal{E}\ar[d]^{\rd_\phi}\\ 
    (\mathcal{T}\mathcal{E})_\Phi=
    \mathcal{T}(\mathcal{E}_\phi)\ar[r]^-{\tc_{\mathcal{E}_\phi}}& 
    \mathcal{E}_\phi
  }
\end{equation*}
is commutative. An immediate consequence of this fact is

\begin{proposition}
  Let~$\mathcal{R}\subset
  \mathcal{T}\mathcal{E}\otimes_{\mathcal{E}}\mathcal{T}\mathcal{E}$ be a
  recursion operator for symmetries of equation~$\mathcal{E}$ and~$\phi$ be a
  symmetry of~$\mathcal{E}$. If~$\mathcal{R}$ is invariant with respect
  to~$\phi$ then~$\mathcal{R}_\Phi$ is a recursion operator for symmetries
  of~$\mathcal{E}_\phi$.
\end{proposition}

\subsection{Recursion operators for symmetries of 3D systems}
\label{sec:recurs-oper-symm}

We briefly recall here the results on recursion operators for symmetries of
Equation~\eqref{eq:4}, \eqref{eq:9}, and~\eqref{eq:11} obtained
in~\cite{Morozov2012,Morozov2014}

\subsubsection*{The universal hierarchy equation}
\label{sec:equation-eqrefeq:1-2}

Equation~\eqref{eq:4} admits the following recursion operator
\begin{equation}\label{eq:25}
  \begin{array}{ll}
  D_y(\tilde{\varphi}) &= u_y \, D_x(\varphi) - u_{xy}\, \varphi,   \\[5pt]
  D_z(\tilde{\varphi}) &= u_z \,D_x(\varphi)-D_y(\varphi) - u_{xz}\,\varphi
\end{array}
\end{equation}
that acts on its symmetries.

\subsubsection*{The 3DrdDym equation}
\label{sec:equation-eqrefeq:2-2}
The B\"{a}cklund transformation
\begin{equation}\label{eq:26}
  \begin{array}{ll}
  D_x(\tilde{\varphi}) &= u_x \, D_x(\varphi) - D_t(\varphi) - u_{xx}\,
  \varphi,   \\[5pt]
  D_y(\tilde{\varphi}) &= u_y \,D_x(\varphi) - u_{xy}\,\varphi
\end{array}
\end{equation}
is a recursion operator for symmetries of Equation~\eqref{eq:9}.

\subsubsection*{The Pavlov equation}
\label{sec:equation-eqrefeq:3-2}
The relations
\begin{equation}\label{eq:24}
  \begin{array}{ll}
  D_x(\tilde{\varphi}) &= u_x \, D_x(\varphi) + D_y(\varphi) - u_{xx}\,
  \varphi,   \\[5pt]
  D_y(\tilde{\varphi}) &= D_t(\varphi)+ u_y \,D_x(\varphi) - u_{xy}\,\varphi.
\end{array}
\end{equation}
are a recursion operator for symmetries of Equation~\eqref{eq:11}.

\subsection{The negative result}
\label{sec:negative-result}
Here we show that the general construction of
Section~\ref{sec:general-construction-1} produces no recursion operator for
the reduced equations under consideration.
\begin{proposition}
  Recursion operators~\eqref{eq:25}\textup{,} \eqref{eq:26} and~\eqref{eq:24}
  are not invariant with respect to the natural lifts of the
  symmetries~\eqref{uhe_symmetry}\textup{,}~\eqref{3DrdDym_symmetry}\textup{,}
  and~\eqref{Pavlov_symmetry}\textup{,} respectively.
\end{proposition}

\begin{proof}
  By direct check.
\end{proof}

\begin{remark}
  The same fact holds for the reduction of the Pavlov equation that leads to
  the Gibbons-Tsarev equation.
\end{remark}

\section{Discussion}
\label{sec:conclusions}

Let us first establish the following fact:
\begin{proposition}
  Equations~\eqref{eq:1}\textup{,} \eqref{eq:2}\textup{,} and~\eqref{eq:3} are
  pair-wise inequivalent.
\end{proposition}
\begin{proof}
  Let us first compare dimensions (see~Table~\ref{tab:1}).
  \begin{table}[hbm]
    \centering
    \begin{tabular}{|l|c|c|}\hline
      &$\dim\sym(\mathcal{E})$&$\dim\cosym(\mathcal{E})$\\\hline
      Equation~\eqref{eq:1}&$4$&$6$\\
      Equation~\eqref{eq:2}&$4$&$4$\\
      Equation~\eqref{eq:3}&$4$&$6$\\
      \hline
    \end{tabular}\vspace{3pt}
    \caption{Dimensions of symmetry and cosymmetry spaces}
    \label{tab:1}
  \end{table}
  Consequently, only Equations~\eqref{eq:1} and~\eqref{eq:3} may be
  equivalent. Now, the Lie algebra structure of~$\sym(\mathcal{E})$
  \begin{table}[hbm]
    \centering
    \begin{tabular}{|c|ccc|}\hline
      Eq.~\eqref{eq:1}&$\phi_0$&$\phi'_0$&$\phi_1$\\\hline
      $\phi_{-1}$&$\phi_{-1}$&$0$&$0$\\
      $\phi_0$&$*$&$0$&$\phi_1$\\
      $\phi'_0$&$*$&$*$&$-\phi_1$\\
      \hline
    \end{tabular}\qquad\qquad
    \begin{tabular}{|c|ccc|}\hline
      Eq.~\eqref{eq:3}&$\phi_{-2}$&$\phi_{-1}$&$\phi_0$\\\hline
      $\phi_{-3}$&$0$&$0$&$-\phi_{-3}$\\
      $\phi_{-2}$&$*$&$-\phi_{-3}$&$\frac{2}{3}\phi_{-2}$\\
      $\phi_{-1}$&$*$&$*$&$-\frac{1}{3}\phi_{-1}$\\
      \hline
    \end{tabular}\vspace{3pt}
    \caption{Commutators in $\sym\mathcal{E}_{\eqref{eq:1}}$ and
      $\sym\mathcal{E}_{\eqref{eq:3}}$} 
    \label{tab:2}
  \end{table}
  for Equations~\eqref{eq:1} and~~\eqref{eq:3} is presented in
  Table~\ref{tab:2}. One can see that dimension of the commutant in the first
  case is~$2$, while in the second case it equals~$3$. Thus, the algebras are
  not isomorphic.
\end{proof}

\begin{remark}
  The equations under consideration are not equivalent to the Gibbons-Tsarev
  equation, because the symmetry algebra of the latter is five-dimensional.
\end{remark}

Nevertheless, as we saw, all these equations have several common features. In
particular, we would like
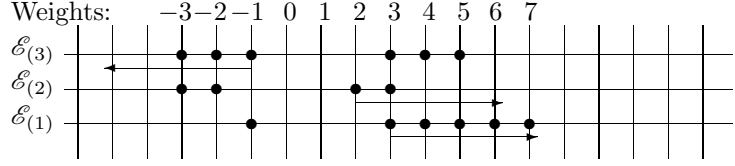
\begin{figure}
  \centering
  \begin{picture}(260,50)\linethickness{.01pt}
    \put(0,52){Weights:}
    \put(0,13){$\mathcal{E}_{(1)}$}
    \put(0,26){$\mathcal{E}_{(2)}$}
    \put(0,39){$\mathcal{E}_{(3)}$}
    \multiput(25,0)(13,0){19}{\line(0,1){50}}
    \put(20,13){\line(1,0){250}}
    \put(20,26){\line(1,0){250}}
    \put(20,39){\line(1,0){250}}
    \put(55,52){$-3$}
    \put(68,52){$-2$}
    \put(82,52){$-1$}
    \put(102,52){$0$}
    \put(115,52){$1$}
    \put(128,52){$2$}
    \put(141,52){$3$}
    \put(154,52){$4$}
    \put(167,52){$5$}
    \put(179,52){$6$}
    \put(192,52){$7$}
    \put(64,39){\circle*{4}}
    \put(77,39){\circle*{4}}
    \put(90,39){\circle*{4}}
    \put(90,34){\vector(-1,0){55}}
    \put(142,39){\circle*{4}}
    \put(155,39){\circle*{4}}
    \put(168,39){\circle*{4}}
    \put(64,26){\circle*{4}}
    \put(77,26){\circle*{4}}
    \put(129,26){\circle*{4}}
    \put(129,21){\vector(1,0){55}}
    \put(142,26){\circle*{4}}
    \put(90,13){\circle*{4}}
    \put(142,13){\circle*{4}}
    \put(142,8){\vector(1,0){55}}
    \put(155,13){\circle*{4}}
    \put(168,13){\circle*{4}}
    \put(181,13){\circle*{4}}
    \put(194,13){\circle*{4}}
  \end{picture}
  \caption{Distribution of cosymmetries}
  \label{fig:1}
\end{figure}
to indicate how local cosymmetries of our equations are distributed with
respect to weights (see Figure~\ref{fig:1}). In all three cases, they fit into
two disjoint groups with certain gaps between them: the first one consist of
cosymmetries whose corresponding conservation laws are members of infinite
series (these are underlined by arrows, and the arrow itself indicates the
direction to which the sequence of conservation laws goes). The second group
includes `standing-alone' cosymmetries.

\begin{remark}
  A similar picture is observed in the case of the Gibbons-Tsarev equation. It
  also possesses a `standing-alone' cosymmetry of order three.
\end{remark}

A natural question arises: does there exist a construction, similar to the one
of Section~\ref{sec:hier-nonl-cons}, that allows to embed the conservation laws
corresponding to the `standing-alone' cosymmetries into other infinite
hierarchies?

Another question relates to the algebras of nonlocal symmetries in the
infinite-dimensional coverings constructed above. It seems that such an
algebra for Equation~\eqref{eq:3} should be similar (or isomorphic to that of
the Gibbons-Tsarev equation), while the algebras for Equations~\eqref{eq:1}
and~\eqref{eq:2} are different: all these Lie algebras are graded, but in the
first two cases all homogeneous components are one-dimensional and for other
equations this is not the case.

Finally, it is interesting to study the structure of symmetries and
cosymmetries of the reductions that admit symmetry algebras with functional
parameters (see the Introduction) and compare them with the results described
here.

All these problems are subject to future research.

\section{Appendix: Conservation laws}
\label{sec:append-cons-laws}

We present here the conservation laws that correspond to the cosymmetries
described above. Everywhere below~$\abs{\omega_i}=i$. We also use the
notation~$\psi_\omega\in\cosym(\mathcal{E})$ for the generating function of a
conservation law~$\omega$.

\subsubsection*{Equation~\eqref{eq:1}}
\label{sec:conservation-laws}
The space of corresponding conservation laws is $6$-dimensional and spans the
following elements~$\omega_i=P_i\,dx+Q_i\,dy$:
\begin{align*}
  P_{-4} & = e^{4 x} (u_x^2 u_y+8 u^3 u_x+13 u^2 u_x^2+2 u u_x^3+8 u^2
  u_y+u_y^2-3 u u_x^2 u_{xx}+2 u u_x u_y\\
  &-2 u u_x u_{xy}-2 u u_{xx} u_y),\\
  Q_{-4} &= u e^{4 x} (-2 u_x u_y u_{xx}+3 u_x^2 u_y-u_x^2 u_{xy}+8 u u_x u_y+2 u
  u_x u_{xy}+4 u_y^2\\
  &-2 u_y u_{xy});\\[3pt]
  P_{-3} &= e^{3 x} (-u u_{xy}+u_x u_y+3 u^2 u_x+u u_x^2-2 u u_x u_{xx}),\\
  Q_{-3} &= u e^{3 x} (-u_y u_{xx}-u_x u_{xy}+u u_{xy}+2 u_x u_y);\\[3pt]
  P_{-2} &= -e^{2 x} (-u_y+u u_x+u u_{xx}),\\
  Q_{-2} &= -u e^{2 x} u_{xy};\\[3pt]
  P_2&=-\frac{1}{u_y},\\
  Q_2&=\frac{1}{u_y}(u_x+u);\\[3pt]
  P_3 &= \frac{1}{u_y^3}(u_y^2 y+2 u u_{xy}-u u_y-2 u_x u_y),\\
  Q_3 &= -\frac{1}{u_y^3}(u u_y^2 y+u_x u_y^2 y+2 u^2 u_{xy}-u^2 u_y+2 u u_x
  u_{xy}-4 u u_x u_y-2 u u_{xx} u_y\\
  &-u_x^2 u_y);\\[3pt]
  P_4 &= \frac{1}{u_y^4}(-u_y^3 y^2-4 u u_{xy} u_y y+2 u u_y^2 y+4 u_x u_y^2
  y-u^2 u_y+6 u u_x u_{xy}\\
  &-2 u u_x u_y-2 u u_{xx} u_y-3 u_x^2 u_y-u_y^2),\\
  Q_4 &=\frac{1}{u_y^4}(u u_y^3 y^2+u_x u_y^3 y^2+4 u^2 u_{xy} u_y y-2 u^2 u_y^2
  y+4 u u_x u_{xy} u_y y\\
  &-8 u u_x u_y^2 y-4 u u_{xx} u_y^2 y-2 u_x^2 u_y^2 y+u^3 u_y-6 u^2 u_x
  u_{xy}+3 u^2 u_x u_y\\
  &-6 u u_x^2 u_{xy}+9 u u_x^2 u_y+6 u u_x u_{xx} u_y +u_x^3 u_y-2 u u_{xy}
  u_y+4 u u_y^2). 
\end{align*}
Here~$\abs{\psi_{\omega}}=\abs{\omega}+1$.

\subsubsection*{Equation~\eqref{eq:2}}
\label{sec:conservation-laws-1}

The space of conservation laws is $4$-dimensional and is generated
by~$\omega_i=P_i\,dx+Q_i\,dy$ of the form
\begin{align*}
  P_{-2} &= \frac{1}{2} (2 u u_{xy}-2 u_x u_y-u_y x) \frac{e^{-2 y}}{u_y^3},\\
  Q_{-2} &= \frac{1}{2} (2 u u_x u_{xy}-2 u u_{xx} u_y+2 u u_{xy} x-u_x^2
  u_y-2 u_x u_y
  x-u_y x^2-2 u u_y) \frac{e^{-2 y}}{u_y^3};\\[3pt]
  P_{-1} &= -\frac{e^{-y}}{u_y},\\
  Q_{-1} &= -(u_x+x)\frac{ e^{-y}}{u_y};\\[3pt]
  P_3 &= u u_{xx}+3 u+u_y,\\
  Q_3 &= u u_{xy}+u_y x;\\[3pt]
  P_4 &=- \frac{1}{2} u u_{xy}+2 u_y x+ \frac{1}{2} u_x u_y +  \frac{5}{2} u x
  u_{xx}+u u_x u_{xx}+8 u x+ \frac{1}{2} u u_x,\\ 
  Q_4 &= 2 u_y x^2+ \frac{1}{2} u_x u_y x+2 u u_{xy} x+ \frac{1}{2} u u_x u_{xy}+
  \frac{1}{2} u u_{xx}  u_y+u u_y. 
\end{align*}
Again,~$\abs{\psi_\omega}=\abs{\omega}-1$.

\subsubsection*{Equation~\eqref{eq:3}}
\label{sec:conservation-laws-2}

The space of conservation laws is $6$-dimensional; elements
$\omega_i=P_i\,dx+Q_i\,dy$ of a basis are
\begin{align*}
  P_8&=u_y^3 u_x u_{yy} u + \frac{1}{5} u x u_y^3 u_{xy} + \frac{116}{5} u x^2
  u_x u_{xy} + \frac{162}{5} u x u_x u_y + \frac{229}{15} u x^3 u_y u_{xy}\\
  & + \frac{8}{5} u x^2 u_y^2 u_{xy} + \frac{3}{5} u_y^2 u_x u_{xy} u+
  \frac{379}{15} u_x u_{yy} u x^3 +
  \frac{758}{15} u_{yy} u x^3 y + 2 u_y^3 u_{yy} u y\\
  & + \frac{184}{5} u_{yy} u x y^2 + \frac{348}{5} u x^2 y u_{xy} -
  \frac{48}{5} x y u_x u_y^2+\frac{6}{5} u y u_y^2 u_{xy} + \frac{72}{5} u_y
  u_{yy} u y^2\\
  & + \frac{12}{5} u_y u_x^2 u_{yy} u + 80 u x y u_y + \frac{36}{5} u y u_x
  u_{xy} - \frac{164}{5} x^2 y u_x u_y - \frac{6}{5} y u_x u_y^3 - \frac{8}{5}
  x^2 u_x u_y^3\\
  &- \frac{1024}{15} x^4 y u_y + 43 u x^3 u_y + \frac{48}{5} u y
  u_y^2 + \frac{18}{5} u u_x u_y^2 - \frac{164}{5} x y^2 u_y^2\\
  & - \frac{1}{5} x u_x u_y^4 + \frac{52}{5} u y^2 u_{xy} + \frac{14}{5} x^2
  u_x^2 u_y-\frac{64}{5} x^2 y u_y^3 +
  \frac{2048}{5} u x^2 y + 2 y u_x^2 u_y + \frac{16}{5} u x u_y^3\\
  & + \frac{82}{5} u y u_x + \frac{32}{5} u_x u_y^2 u_{yy} u x + \frac{64}{5}
  u_y^2 u_{yy} u x y + 24 u x y u_y u_{xy}+\frac{132}{5} u_x u_{yy} u x y\\
  &+ 12 u_x u_y u_{yy} u y+ \frac{96}{5} u_x u_y u_{yy} u x^2 + \frac{192}{5}
  u_y u_{yy} u x^2 y + \frac{56}{5} u x u_x u_y u_{xy} + \frac{1}{5} u_x^2
  u_y^3 \\
  &+ \frac{3}{5} u_x^3 u_y+\frac{256}{5} u y^2 + \frac{4096}{15} u x^4 -
  \frac{241}{5} u^2 x + \frac{2}{5} u u_y^4 - \frac{24}{5} y^2 u_y^3 -
  \frac{64}{5} y^3 u_y - \frac{2}{5} y u_y^5 \\
  &+ \frac{8}{5} u u_x^2 - \frac{512}{5} x^2 y^2 u_y + \frac{64}{5} u x^2
  u_y^2+\frac{113}{5} u x^2
  u_x - \frac{512}{15} x^3 y u_y^2 - \frac{16}{5} x y u_y^4 - 4 y^2 u_x u_y\\
  &- \frac{32}{5} x^3 u_x u_y^2 + \frac{6}{5} u_x^2 u_{xy} u - \frac{256}{15}
  x^4 u_x u_y + \frac{127}{3} u x^4 u_{xy}+x u_x^2
  u_y^2,\\
  Q_8&=\frac{36}{5} u y u_x u_{yy} - \frac{72}{5} x y u_x u_y + \frac{42}{5} u
  y u_y^2 u_{yy} + \frac{92}{5} u x y u_{xy} + \frac{32}{5} u x u_y^2 u_{xy} +
  4 u x
  u_x u_{xy}\\
  & + \frac{256}{5} u x^2 y u_{yy} + \frac{12}{5} u_x u_{xy} u_y
  u+\frac{64}{3} u x^3 u_y u_{yy} +
  \frac{72}{5} u x^2 u_y^2 u_{yy} + \frac{36}{5} u y u_y u_{xy} \\
  &+ \frac{28}{5} u x u_y^3 u_{yy} + \frac{96}{5} u x^2 u_x u_{yy} +
  \frac{96}{5} u x^2 u_y u_{xy}+3 u_y^2 u_x u_{yy} u + \frac{52}{5} u y^2
  u_{yy}\\
  &+ \frac{6}{5} u_x^2 u_{yy} u + u_y^4 u_{yy} u + \frac{256}{15} u x^4 u_{yy}
  + \frac{32}{5} x y^2 u_y + \frac{94}{5} u y u_y + \frac{379}{15} u
  x^3 u_{xy}\\
  &- \frac{256}{5} x^2 y u_x+\frac{82}{5} u x u_x + 16 u x u_y^2 -
  \frac{17}{5} x u_x^2 u_y + u_y^3
  u_{xy} u + \frac{32}{5} u u_x u_y \\
  &- \frac{133}{15} x^3 u_x u_y + \frac{256}{5} x^3 y u_y + \frac{512}{5} u x
  y+\frac{176}{5} u x y u_y u_{yy} + \frac{64}{5} u x u_x u_y u_{yy} +
  \frac{12}{5} u u_y^3 \\
  & + \frac{2048}{15} u x^3+ \frac{512}{15} x^5 u_y-2 y u_x^2 - \frac{32}{5}
  y^2 u_x - \frac{41}{5} x^2 u_x^2 - \frac{512}{15} x^4
  u_x-\frac{1}{5} u_x^3;\\[3pt]
  P_7&=\frac{13}{4} u y u_x u_{yy} - \frac{25}{4} x y u_x u_y + 2 u y u_y^2
  u_{yy} + \frac{65}{4} u x y u_{xy} + \frac{1}{4} u x u_y^2 u_{xy} +
  \frac{23}{4} u x u_x u_{xy} \\
  &+ \frac{65}{4} u x^2 y u_{yy} + \frac{3}{2} u_x u_{xy} u_y u+\frac{13}{4} u
  y u_y u_{xy} + \frac{65}{8} u x^2 u_x u_{yy} + \frac{47}{8} u x^2 u_y
  u_{xy}\\
  &+ u_y^2 u_x u_{yy} u + \frac{9}{2} u y^2 u_{yy} + \frac{1}{2} u_x^2 u_{yy}
  u - \frac{49}{2} x y^2 u_y+\frac{45}{4} u y u_y+\frac{391}{24} u x^3 u_{xy}
  + 2 u x u_x \\
  &+ \frac{7}{2} u x u_y^2 + \frac{5}{4} x u_x^2 u_y +\frac{9}{2} u u_x u_y -
  \frac{49}{8} x^3 u_x u_y - \frac{343}{12} x^3 y u_y + \frac{343}{4} u x y -
  \frac{1}{4} x u_x u_y^3\\
  &-y u_x u_y^2 + \frac{21}{2} u x y u_y u_{yy} + \frac{21}{4} u x u_x u_y
  u_{yy}
  + \frac{1}{2} u u_y^3 + \frac{2401}{24} u x^3 \\
  &- \frac{9}{2} y^2 u_y^2 + \frac{1}{4} u_x^2 u_y^2 - \frac{1}{2} y u_y^4 -
  \frac{53}{8} u^2 -\frac{7}{2} x y u_y^3-\frac{49}{4} x^2 y u_y^2 -
  \frac{7}{4} x^2 u_x u_y^2 + \frac{131}{8} u x^2
  u_y,\\
  Q_7&=\frac{21}{4} u x u_x u_{yy} + \frac{21}{4} u x u_y u_{xy} +
  \frac{11}{2} u y u_y u_{yy} + \frac{35}{4} u x^2 u_y u_{yy}+\frac{9}{2} u x
  u_y^2 u_{yy} \\
  &+ 2 u_y u_x u_{yy} u + 14 u x y u_{yy} - \frac{49}{4} x y u_x+\frac{343}{8}
  u x^2 + \frac{49}{4} u y - 2 x u_x^2 - \frac{1}{2} u_x^2 u_y -
  \frac{343}{24} x^3 u_x \\
  &+ 2 u u_x + \frac{343}{24} x^4 u_y + \frac{5}{2} u u_y^2 + \frac{1}{2} u_x
  u_{xy} u + \frac{49}{4} x^2 y u_y+ \frac{49}{6} u x^3 u_{yy} + \frac{65}{8}
  u x^2 u_{xy} \\
  &- \frac{9}{4} y u_x u_y + \frac{9}{4} u y u_{xy} - \frac{33}{8} x^2 u_x u_y
  + u_y^3 u_{yy} u + u_y^2
  u_{xy} u;\\[3pt]
  P_6&=12 u y+\frac{2}{3} u u_y^2 + 36 u x^2 + \frac{1}{3} u_x^2 u_y -
  \frac{2}{3} y u_y^3 + \frac{7}{3} u x u_x u_{yy} + \frac{14}{3} u x y u_{yy}
  +   u_y u_x u_{yy} u \\
  &+ 2 u y u_y u_{yy} + \frac{17}{3} u x u_y+2 u x u_y u_{xy} - \frac{2}{3} y
  u_x u_y+\frac{8}{3} u y u_{xy}+u_x u_{xy}
  u + \frac{19}{3} u x^2 u_{xy}\\
  &-12 x^2 y u_y-2 x^2 u_x u_y-4 x y u_y^2 -\frac{1}{3}
  x u_x u_y^2 - 4 y^2 u_y-\frac{1}{3} u u_x,\\
  Q_6&=12 x u-6 x^2 u_x-2 y u_x+6 x^3 u_y-\frac{1}{3} u_x^2 + \frac{10}{3} u x
  u_y u_{yy} - \frac{5}{3} x u_x u_y + 4 u x^2 u_{yy} \\
  &+ u_y^2 u_{yy} u +
  \frac{7}{3} u x u_{xy} + \frac{8}{3} u y u_{yy}+u_x u_{yy} u+u_{xy} u_y u+2
  x u_y y;\\[3pt] 
  P_5&=-5 x u_y y - \frac{5}{2} x u_x u_y - u_y^2 y - \frac{1}{2} u_y^2 u_x +
  \frac{25}{2} x u + \frac{1}{2} u_x u_{yy} u + u y u_{yy} - \frac{1}{2}
  u_{xy} u_y u \\
  &+ \frac{1}{2} u x u_{xy} - \frac{1}{2} u u_y,\\
  Q_5&=\frac{1}{2} u u_{xy} + \frac{5}{2} u - \frac{5}{2} x u_x - \frac{1}{2}
  u_y u_x + \frac{5}{2} x^2 u_y-2 x u_y^2-\frac{1}{2} u_y^3;\\[3pt]
  P_4&=-u_y u_x-2 u_y y+4 u,\\
  Q_4&= -u_y^2+x u_y-u_x;\\[3pt]
  P_0&=\frac{u_y}{x u_y-u_x-2 y},\\
  Q_0&=-\frac{1}{x u_y-u_x-2 y}.
\end{align*}
Here~$\abs{\psi_\omega}=\abs{\omega}-1$.

\section*{Acknowledgements}
\label{sec:acknowledgements}
The authors are grateful to E.~Ferapontov for remarks and discussion.  The
$2^{\mathrm{nd}}$ author is grateful to the Mathematical Institute of the
Silesian University in Opava for support and comfortable working condition.
Computations of symmetry algebras were fulfilled using the \textsc{Jets}
software,~\cite{B-M:Jets}.


\begin{thebibliography}{9}
\bibitem{B-K-M-V:Sym-Red} H.~Baran, I.S.~Krasil{\cprime}shchik, O.I.~Morozov,
  P.~Voj{\v{c}}{\'{a}}k, \emph{Symmetry reductions and exact solutions of Lax
    integrable $3$-dimensional systems}, Journal of Nonlinear Mathematical
  Physics, Vol. 21, No. 4 (December 2014), 643--671; \url{arXiv:1407.0246}
  [nlin.SI], DOI: 10.1080/14029251.2014.975532,
  \url{http://www.tandfonline.com/doi/abs/10.1080/14029251.2014.975532#.VE4v8Bbim_k}.
\bibitem{B-M:Jets} H.~Baran, M.~Marvan, \emph{Jets. A software for
    differential calculus on jet spaces and diffeties}.
  \url{http://jets.math.slu.cz}.
\bibitem{Liou} F.~Calogero, A.~Degasperis, \emph{Spectral Transform and
    Solitons: Tools to Solve and Investigate Nonlinear Evolution
    Equations}. New York: North-Holland, p.~60, 1982.
\bibitem{GT} J.~Gibbons,  S.P.~Tsarev, \emph{Reductions of the Benney
    equations}, Phys.\ Lett.\ \textbf{A 211} (1996) 19--24.
\bibitem{I.K:Nat-Constr} I.S.~Krasil{\cprime}shchik, \emph{A natural geometric
    construction underlying a class of Lax pairs}, Lobachevskii Journal of
  Mathematics, 2015 (to appear).
\bibitem{KLV} I.S.~Krasil{\cprime}shchik, V.V.~Lychagin, A.M.~Vinogradov,
  \emph{Geometry of Jet Spaces and Nonlinear Differential Equations}, Adv.\
  Stud.\ Contemp.\ Math.~1, Gordon and Breach, New York, London, 1986.
\bibitem{KV} I.S.~Krasil'shchik, A.M.~Vinogradov. \emph{Nonlocal trends in the
    geometry of differential equations: Symmetries, conservation laws, and
    B\"{a}cklund transformations}.  Acta Applicandae Mathematica, \textbf{15}
  (1989) no.~1-2, pp.~161--209.
\bibitem{K-V:GomJetSpaces} I.S.~Krasil{\cprime}shchik, A.M.~Verbovetsky,
  \emph{Geometry of jet spaces and integrable systems} Journal of Geometry and
  Physics, \textbf{61} (2011) Issue 9, 1633--1674,
  \url{arXiv:1002.0077 [math.DG]}.
\bibitem{Mar:Another} M.~Marvan, \emph{Another look on recursion operators},
  in: Differential Geometry and Applications, Proc.\ Conf.\ Brno, 1995
  (Masaryk University, Brno, 1996) 393--402.
\bibitem{Morozov2012} O.I.~Morozov, \emph{Recursion Operators and Nonlocal
    Symmetries for Integrable rmdKP and rdDym Equations},
  \url{arXiv:1202.2308}, 2012
\bibitem{Morozov2014} O.I. Morozov, \emph{A recursion operator for the
    universal hierarchy equation via Cartan's method of equivalence}, Central
  European Journal of Mathematics, {\bf 12 (2)}, 2014, 271--283
\bibitem{O-S:NonHom-Sys} A.V.~Odesskii, V.V.~Sokolov, \emph{Non-homogeneous
    systems of hydrodynamic type possessing Lax representations},
  \url{arXiv:1206.5230}, 2006.
\bibitem{O-S:Sys-G-T} A.V.~Odesskii, V.V.~Sokolov, \emph{Systems of
    Gibbons-Tsarev type and integrable $3$-dimensional models},
  \url{arXiv:0906.3509}, 2009.
\bibitem{P-J-Y:Int-M-S-H} M.V.~Pavlov, Jen Hsu Chang, Yu Tung Chen,
  \emph{Integrability of the Manakov-Santini hierarchy},
  \url{arXiv:0910.2400}, 2009.
\end{thebibliography}
\end{document}